\documentclass[11pt]{article}
\usepackage{setspace}
\usepackage{amsmath}
\usepackage{geometry}
\usepackage{amssymb}
\usepackage{amsmath}
\usepackage{graphicx} 
\usepackage{amsthm}
\usepackage{microtype}
\usepackage[colorlinks=true]{hyperref}
\usepackage{tikz}
\usepackage{authblk}
\usepackage{url}
\usepackage{algorithmic}
\usepackage{algorithm}
\newtheorem{lm}{Lemma}
\newtheorem{thm}[lm]{Theorem}
\theoremstyle{definition}
\newtheorem{defi}[lm]{Definition}

\newtheorem{qq}[lm]{Problem}
\newtheorem{eg}[lm]{Example}

\renewcommand{\(}{\left(}
\renewcommand{\)}{\right)}

\DeclareMathOperator{\capa}{cap}

\title{Mitigating Overexposure in Viral Marketing}
\author[1]{Rediet Abebe}
\author[2]{Lada A. Adamic}
\author[1]{Jon Kleinberg}
\affil[1]{Cornell University}
\affil[2]{University of Michigan}
\date{}                     
\begin{document}
\maketitle

\newcommand{\omt}[1]{}
\newcommand{\xhdr}[1]{\vspace{1.7mm}\noindent{{\bf #1.}}}

\begin{abstract}
In traditional models for word-of-mouth recommendations and
viral marketing, the objective function has generally been based
on reaching as many people as possible.
However, a number of studies have shown that the indiscriminate
spread of a product by word-of-mouth can result in 
{\em overexposure}, reaching people who evaluate it negatively.
This can lead to an effect in which the over-promotion of a product
can produce negative reputational effects, by reaching a part of
the audience that is not receptive to it.

How should one make use of social influence when there is a risk
of overexposure?
In this paper, we develop and analyze a theoretical model for this process;
we show how it captures a number of the qualitative phenomena associated
with overexposure, and for the main formulation of our model, we provide
a polynomial-time algorithm to find the optimal marketing strategy.
We also present simulations of the model on real network topologies,
quantifying the extent to which our optimal strategies outperform
natural baselines. 

\end{abstract}

\section{Introduction}\label{intro}

A rich line of research has studied the effectiveness of marketing
strategies based on person-to-person recommendation within a 
social network --- a process often termed {\em viral marketing}
\cite{jurvetson-viral}
and closely connected to the broader sociological literature on
the {\em diffusion of innovations} in social networks \cite{rogers-diffusion}.
A key genre of theoretical question that emerged early in this literature
is the problem of optimally ``seeding'' a product 
in a social network through the selection of a set 
of initial adopters \cite{dr,kkt,rd}.
In this class of questions, we consider a firm that
has a product they would like to market to a group
of agents on a social network;
it is often the case that the
firm cannot target all the participants in the network, and 
so they seek to target the most influential ones so as to maximize
exposure and create a cascade of adoptions. 
Approaches to this question have generally been based on objective
functions in which the goal is to maximize the number of people
who are reached by the network cascade --- or more generally, in 
which the objective function monotonically increases in the
number of people reached.

\xhdr{The dangers of overexposure}
Separately from this, lines of research in both marketing
and in the dynamics of on-line information have provided diverse evidence
that the benefits of a marketing campaign are not in fact purely increasing
in the number of people reached.
An influential example of such a finding is the {\em Groupon effect},
in which viral marketing via Groupon coupons leads to lower Yelp
ratings. In \cite{group}, they note the negative effect Groupon has on
average Yelp ratings and provide arguments for the underlying
mechanism; one of their central hypotheses is that by using Groupon
as a matchmaker, businesses may be attracting customers from a portion
of the population that is less inclined to like the product. 
In another example, Kovcs and Sharkey \cite{good}  discuss a setting
on Goodreads where books that win 
prestigious awards (or are short-listed for them) 
attract more readers following the announcement, which
again leads to a drop in the average rating of the book on the platform.
Aizen et al. \cite{aizen-batting} show a similar effect for on-line
videos and other media; they receive a discontinuous drop in their ratings
when a popular blog links to them, driving users to the item who
may not be interested in it.

Research in marketing has shown that exposure to different groups and
influence between such groups can help or hurt adoption \cite{berger2007consumers,hb,turf}. 
For example, Hu and Van den Bulte \cite{hb} argue that 
agents adopt products to boost their status;
and so as word-of-mouth effects for a product become stronger
among middle-income individuals, 
there might be a negative impact on adoption among 
higher-income individuals. Similar behavior is observed in health campaigns. 
In \cite{health}, Wakefield et al. discuss the importance of segmenting populations 
and exposing groups to anti-smoking campaigns whose themes the group 
is most susceptible to in order to maximize the impact of future campaigns. 

We think of these effects collectively as different forms of
{\em overexposure}; while reaching many potential customers is
not a concern in and of itself, the empirical research above suggests that
there may exist particular subsets of the 
population --- potentially large subsets --- 
who will react negatively to the product.
When a marketing cascade reaches members of this negatively inclined
subset, the marketing campaign can suffer negative payoff that may
offset the benefits it has received from other parts of the population.
This negative payoff can come in the form of harm to the
firm's reputation, either through latent 
consumer impressions and effects on 
brand loyalty \cite{reputation2,reputation1}
or through explicitly visible negative reviews on rating sites.

Despite the importance of these considerations in marketing, 
they have not been incorporated into models of influence-based
marketing in social networks.
What types of algorithmic issues arise when we seek to 
spread a word-of-mouth cascade through a network, but must
simultaneously ensure that it reaches the ``right'' part of the
audience --- the potential customers who will like the product,
rather than those who will react negatively to it?

\xhdr{The present work: A model of cascades with the risk of overexposure}
In this paper, we propose a basic theoretical model for 
the problem of seeding a cascade when there are benefits from
reaching positively inclined customers and costs from reaching
negatively inclined customers.

There are many potential factors that play a role in the distinction
between positively and negatively inclined customers, 
and for our model we focus on a stylized framework in which each product
has a known parameter $\phi$ in the interval $[0,1]$
that serves as some measure for the breadth of its appeal.
At this level of generality, this parameter could serve as a proxy for a number of things,
including quality; or a one-dimensional combination of price and quality;
or --- in the case where the social network represents a population
defined by a specific interest --- compatibility with the core interests 
of network's members.

Each node in the network is an {\em agent} who will evaluate the
product when they first learn of it; agents differ in how 
{\em critical} they are of new products, with agents of low criticality
tending to like a wider range of products and agents of high criticality 
tending to reject more products.
Thus, each agent $i$ has a criticality parameter
$\theta_i$ in the interval $[0, 1]$; since we assume that
the firm has a history of marketing products to this network 
over a period of time, it knows this parameter $\theta_i$.
When exposed to a product, an agent accepts the product if 
$\phi \geq \theta_i$ and advertises the
product to their neighbors, 
leading to the potential for a cascade.
However, if $\phi < \theta_i$, then the agent rejects the product,
which results in a negative payoff to the firm; the cascade stops
at such agents $i$, since they do not advertise it to their neighbors.

The firm's goal is to advertise the product to a subset of the nodes
in the network --- the {\em seed set} --- 
resulting in a potential cascade of further
nodes who learn about the product, so as to maximize its overall payoff.
This payoff includes a positive term for each agent $i$ who sees
the product and has $\phi \geq \theta_i$,
and a negative term for each agent $i$ who sees
the product and has $\phi < \theta_i$;
agents who are never reached by the cascade never find out about
the product, and the firm gets zero payoff from them.

\xhdr{Overview of Results}
We obtain theoretical results for two main settings of this problem: the 
{\em unbudgeted case}, in which the firm can initially advertise the
product to an arbitrary seed set of nodes, and the
{\em budgeted case}, in which the firm can advertise the project
to at most $k$ nodes, for a given parameter $k$.
We note that typically in influence maximization problems, the
unbudgeted case is not interesting: if the payoff is monotonically
increasing in the number of nodes who are exposed to the product,
then the optimal unbudgeted strategy is simply to show the product
to everyone.
In a world with negative payoffs from overexposure, however, the
unbudgeted optimzation problem becomes non-trivial:
we must tradeoff the benefits of showing the product to customers
who will like it against the negatives that arise when these
customers in turn share it with others who do not. 

For the unbudgeted problem, we give a polynomial-time algorithm for finding
the optimal seed set.  The algorithm uses network flow techniques
on a graph derived from the underlying social network with
the given set of parameters $\theta_i$.
In contrast, we provide an NP-hardness result for 
the budgeted problem.

We then provide a natural generalization of the model:
rather than each agent exhibiting only two possible behaviors
(rejecting the product, or accepting it and promoting it),
we allow for a wider range of agent behaviors.
In particular, we will assume each agent has 
three parameters which control 
whether the agent ignores the product, views but rejects the product, 
accepts the product but does not broadcast it to its neighbors, and 
accepts the product and advertises it to neighbors. 
We show how to extend our results to this more general case,
obtaining a polynomial-time algorithm for the unbudgeted case
and an NP-hardness result for the budgeted case.

Finally, we perform computational simulations of our algorithm for
the unbudgeted case on sample network topologies derived from
moderately-sized social networks.
We find an interesting effect in which the performance of the
optimal algorithm transitions between two behaviors as $\phi$ varies.
For small $\phi$ the payoff grows slowly while a baseline that promotes
the product to every agent $i$ with $\theta_i < \phi$ achieves
negative payoff (reflecting the consequences of overexposure).
Then, for large $\phi$, the payoff grows quickly, approaching a simple
upper bound consisting of all $i$ for which $\theta_i < \phi$.

\section{Preliminaries}\label{prelim}

There is a product with a parameter $\phi \in [0, 1]$,
measuring the breadth of its appeal.
$G$ is an unweighted, undirected graph with $n$ agents as 
its nodes.
For each agent $i$, the agent's 
criticality parameter $\theta_i \in [0, 1]$ measures
the minimum threshold for $\phi$ the agent demands before
adoption. Thus, higher values of $\theta_i$ correspond to more critical
agents. We assume that these values are fixed and known to the firm.

The firm chooses an initial set of agents $S \subseteq V$ 
to ``seed'' with the product.  
If an agent $i$ sees the product, it accepts it if $\theta_i \leq \phi$
and rejects it if $\theta_i > \phi$.
We say that an agent $i$ is {\em accepting} in the former case and 
{\em rejecting} in the latter case. 
Each accepting agent who is exposed to 
the product advertises it to their neighbors,
who then, recursively, are also exposed to the product.
We will assume throughout that the firm chooses a seed set consisting
entirely of accepting nodes (noting, of course, that rejecting nodes
might subsequently be exposed to the product after nodes in 
the seed set advertise it to their neighbors).

We write $V(S)$ for the set of agents exposed to the product if 
the seed set is $S$. Formally, $V(S)$ is the set of all agents $i$ who have a
path to some node in $j \in S$ such that all of the internal nodes on the $i$-$j$ path are accepting agents; this is the ``chain of recommendations" 
by which the product reached $i$. 
 Among the nodes in $V(S)$, we define $V^+(S)$ to be the set of agents  who accept the product and $V^-(S)$ to be the set of 
agents who reject the product.

The payoff function associated with seed set $S$ is:
\begin{align}\label{payoff1}
\pi (S) = |V^+(S)|- |V^-(S)|.
\end{align}
We can, more generally, assume that there is a payoff of $p$ to 
accepting the product and a negative payoff of $q$ to rejecting 
the product, and we set the payoff function to be: 
\begin{align}\label{payoff2}
\pi(S) = p |V^+(S)| - q |V^-(S)|.
\end{align}
We will call this the \emph{generalized payoff function}, and simply 
refer to Equation \ref{payoff1} as the \emph{payoff function}. The 
overarching question then is: 

\begin{qq}\label{qq1}
Given a set of agents $V$ with criticality parameters $\theta_i$ ($i \in V$) 
on a social network $G = (V, E)$, and given a product of quality $\phi$, 
what is the optimal seed set $S \subseteq V$ that the firm should target 
in order to maximize the payoff given by Equation (\ref{payoff2})?
\end{qq}


In contrast to much of the influence maximization literature, we assume that
the agents' likelihood of adoption, once exposed to this product, is
not affected by which of their neighbors have accepted or rejected the
product. This differs from, for instance, models in which each
agent requires a certain fraction (or number) of its neighbors to have
accepted the product before it does; or models where probabilistic
contagion takes place across the edges.
These all form interesting directions for further work; here, however,
we focus on questions in which the intrinsic appeal of the product, via
$\phi$, determines adoption decisions, and the social network provides
communication pathways for other agents to hear about the product.

Before proceeding to the main result, we develop some further
terminology that will be helpful in reasoning about the seed sets. 

\begin{defi}
Let $i$ be an accepting node, and let $S = \{i\}$.
Then we say that $V(S)$ is the {\em cluster} of $i$, denoted by $C_i$;
we call $V^+(S)$ the {\em interior} of $C_i$ and denote it by $C_i^o$,
and we call $V^-(S)$ the {\em boundary} of $C_i$ and denote it by $C_i^b$.
\end{defi}

We denote the payoff corresponding to the seed set $S= \{i\}$ by $\pi_i$. 
Note that, 

$$\pi_i = p|C_i^o| - q|C_i^b|.$$

\begin{lm}
Given an accepting node $i \in V$,
and a node $j \in C_i^o$, we have $C_i = C_j$. 
\label{lm:cluster-equiv}
\end{lm} 

\begin{proof}
If $j$ is in the interior of $C_i$, then there exists a path 
$(k_1, k_2, \cdots, k_\ell)$ in $G$, where $i = k_1$ and $j = k_\ell$ such 
that each node along the path has $\theta \leq \phi$. (That is, each node 
$k_i$ is exposed to and accepts the product as a result of $k_{i-1}$'s 
advertisement.) We would like to prove that if $S = \{j\}$, then $i$ would 
be exposed to the product. Equivalently, we want to show 
there exists a path from 
$j$ to $i$ of nodes with $\theta \leq \phi$; but this is precisely the path 
$(k_\ell, k_{\ell - 1}, \cdots, k_1)$. 
\end{proof}

For an arbitrary seed set $S$, the set $V(S)$ may consist of multiple
interior-disjoint clusters, which we label by $\{C_1, C_2, \cdots, C_k\}$, 
where $k \leq |S|$. Note that each of these clusters might be associated with
more than one agent in the seed set and that $\cup_{i = 1}^k C_i = V(S)$. 
(Likewise, $\cup_{i = 1}^k C_i^o = V^+(S)$ and $\cup_{i = 1}^k C_i^b = V^-(S)$.)

Given a seed-set $S$ and corresponding clusters, a direct consequence of Lemma \ref{lm:cluster-equiv} is that adding more nodes 
already contained in these clusters to the seed-set does not change 
the payoff. 

\begin{lm}
Given a set $S'$ of accepting nodes such that $S \subseteq S' \subseteq V(S)$,
we have $\pi(S) = \pi(S')$. 
\end{lm}

It therefore suffices to seed a single agent 
within a cluster. Given a cluster $C_i$, we will simply pick an arbitrary 
node in the interior of the cluster to be the canonical node $i$ and use that 
to refer to the cluster even if $C_i$ is formed as a result of seeding another 
node $j \in C_i^o$.

\section{Main Model}\label{unbudgeted}

Given that all $\theta_i$ are known to the firm, a naive approach would suggest to seed all $i$ where $\pi_i \geq 0$. While this is guaranteed to give a nonnegative payoff, $S$ need not be optimal. 

\begin{eg}\label{counter}
Consider the graph below, where nodes in blue accept the product and those in red reject the product. 

\begin{figure}[!ht]
\centering
\begin{tikzpicture}[scale=1.2,auto=left,every node/.style={circle,fill=gray!20}]\label{fig:countereg}
  \node[fill=blue!50] (n1) at (0, 0.5) {1};
  \node[fill=blue!50] (n2) at (0, 1.5) {2};
  \node[fill=red!50] (n3) at (1.5, 0) {3};
  \node[fill=red!50] (n4) at (1.5, 1) {4};
  \node[fill=red!50] (n5) at (1.5, 2) {5};
  \node[fill=blue!50] (n6) at (3, 0.5) {6};
  \node[fill=blue!50] (n7) at (3, 1.5) {7};

  \foreach \from/\to in {n1/n4,  n1/n2, n1/n5, n1/n3, n2/n4, n2/n5, n2/n3, n3/n6, n3/n7, n4/n6, n4/n7, n5/n6, n5/n7, n6/n7}
    \draw (\from) -- (\to);


\end{tikzpicture}
\end{figure}
Suppose $ p = q = 1$. Then, a naive approach would set $S  = \emptyset$, since each of the resulting clusters to the blue nodes has negative payoff. However, setting $S = \{1, 2, 6, 7\}$ has payoff $1$. 
\end{eg}

This phenomenon is a result of the fact that the clusters $C_i$ might {have} boundaries that intersect non-trivially. Thus, there could be agents whose $\pi_i < 0$, but $C_i^b$ is in a sense ``paid for'' by seeding other agents;
and hence we could have a net-positive payoff from including $i$
subject to seeding other agents whose cluster boundaries intersect with $C_i^b$. Using this observation, we will give a polynomial-time algorithm for finding the optimal seed set under the generalized payoff function using a network flow argument. We first begin by constructing a flow network.

Given an instance defined by $G$ and $\phi$, we let
$\{C_1, C_2, \cdots, C_k\}$ be the set of all distinct clusters in $G$,
with disjoint interiors.
We form a flow network as follows: set $A = \{1,
2, \cdots, k\}$ corresponding to the canonical nodes of the clusters
above and $R$ be the set of agents in the boundaries of all clusters.
We add an edge from
the source node $s$ to each node $i \in A$ with capacity $p \cdot
|C_i^o|$ and label this value by $\capa_i$, and an edge from each node
$j \in R$ to $t$ with capacity $q$. We add an edge between $i$ and $j$
if and only if $j \in C_i^b$, and set these edges to have infinite
capacity. We denote this corresponding flow network by $G_N$.

\usetikzlibrary{arrows}

\begin{eg}

For the above example, $G_N$ is: 

\begin{figure}[!ht]
\centering
\begin{tikzpicture}
[scale=1,auto=left,every node/.style={circle,fill=gray!20}]
\tikzset{edge/.style = {->,> = latex'}}

\node[fill=blue!50] (n1) at (0, 0) {\small{1, 2}};
\node[fill=blue!50] (n2) at (0, 2) {\small{6, 7}};
\node[fill=red!50] (n3) at (2, 0) {3};
\node[fill=red!50] (n4) at (2, 1) {4};
\node[fill=red!50] (n5) at (2, 2) {5};

\node (n7) at (-2, 1) {s};
\node (n8) at (4, 1) {t};

\draw[edge] (n7) to (n1);
\draw[edge] (n7) to (n2);
\draw[edge] (n1) to (n3);
\draw[edge] (n1) to (n4);
\draw[edge] (n1) to (n5);
\draw[edge] (n2) to (n3);
\draw[edge] (n2) to (n4);
\draw[edge] (n2) to (n5);
\draw[edge] (n3) to (n8);
\draw[edge] (n4) to (n8);
\draw[edge] (n5) to (n8);

\end{tikzpicture}
\end{figure}

In this example, the edges out of $s$ have capacities $2p$ and edges into $t$ have capacities $q$. Edges between blue and red notes have infinite capacity. Assuming $4p \geq 3q$, the min-cut $(X,Y)$ has $Y = \{t\}$ and all other nodes in $X$.
\end{eg}


\begin{lm}
Given a min-cut $(X,Y)$ in $G_N$, the optimal seed set in $G$ is
$A \cap X$.
\end{lm}

\begin{proof}
The min-cut $(X,Y)$ must have value at most $q |R|$ since we can
trivially obtain that by setting the cut to be $(V(G_N) \backslash
\{t\}, \{t\})$. 
Given a node $i \in X$, which we recall corresponds to
the canonical node of a cluster, 
if a node $j \in R$ is exposed to the product as a result of
seeding any node in the cluster, then $j \in X$. Otherwise, we would
have edge $(i, j)$ included in the cut, which has infinite capacity
contradicting the minimality of the cut $(X,Y)$. Therefore, the
min-cut will include all nodes in the seed set as well as all nodes
that are exposed to the product as a result of the corresponding
seed set in $S$.

Note that the edges across the cut are of two forms: $(s, i)$ or $(j,
t)$, where $i \in A$ and $j \in R$. The first set of edges contribute
$\sum_{i \in A \cap Y} \capa_i$ (recall $\capa_i = p \cdot | C_i^o|$)
and the latter contributes $| R \cap X| q$. Therefore, the objective
for finding a min-cut can be equally stated as minimizing,
$${\sum_{i \in A \cap Y} \capa_i + |R \cap X| q}.$$
over cuts $(X,Y)$.
Note that $$\sum_{i \in A } \capa_i = \sum_{i \in A \cap X} \capa_i +  \sum_{i \in A \cap Y} \capa_i.$$ Therefore, we have:

\begin{align*}
&\( {\sum_{i \in A \cap Y} \capa_i + |R \cap X| q}\)\\
& = \( { \sum_{i \in A} \capa_i -  \sum_{i \in A \cap X} \capa_i + |R \cap X| q}\)\\
& =  \sum_{i \in A} \capa_i - \( {   \sum_{i \in A \cap X} \capa_i } - |R \cap X| q\)
\end{align*}

Note the term $\( {   \sum_{i \in A \cap X} \capa_i } - |R \cap X| q\)$
is precisely what the payoff objective function is maximizing, giving a correspondence between the min-cut and optimal seed set.
\end{proof}

We therefore have the main result of this section:

\begin{thm}
There is a polynomial-time algorithm for computing the optimal seed set for Problem \ref{qq1} when there are no budgets for the size of the seed set. 
\end{thm}


An interesting phenomenon is that the payoff is not monotone in $\phi$ even when considering optimal seed sets. Take the following example:

\begin{eg}

Suppose we are given the network below with the numbers specifying the $\theta_i$ for each corresponding node:

\begin{figure}[h]
\centering
\begin{tikzpicture}
  [scale=0.9,auto=left,every node/.style={circle,fill=gray!20}]

  \node[fill=blue!50] (n1) at (2, -1) {0.2};
  \node[fill=blue!50] (n2) at (2, 1) {0.2};
 \node[fill=blue!50] (n3) at (4, 0) {0.5};
 \node[fill=blue!50] (n4) at (6, 1) {1};
 \node[fill=blue!50] (n5) at (6, -1) {1};
 \node[fill=blue!50] (n6) at (8, 1) {1};
 \node[fill=blue!50] (n7) at (8, -1) {1};
  \foreach \from/\to in {n1/n2, n1/n3,n2/n3, n3/n4, n3/n5, n3/n6, n3/n7,n4/n5,n4/n6,n4/n7,n5/n6,n5/n7,n6/n7}
    \draw (\from) -- (\to);

\end{tikzpicture}
\end{figure}

If $\phi \in [0, 0.2)$, we cannot do better than the empty-set. If $\phi \in [0.2, 0.5)$, the seed set that includes either of the two left-most nodes gives a payoff of $1$, which is optimal. For the case where $\phi \in (0.5, 1)$, the empty-set is again optimal. 
\end{eg}

This example gives a concrete way to think about overexposure phenomena
such as the Groupon effect \cite{group} discussed in the introduction.
Viewed in the current terms, we could say that by
using Groupon, one could increase the broad-appeal measure of the product (e.g., cheaper, signaling higher quality, etc), which therefore exposes the product to portions of the market that would have previously not been exposed to it, and this could lead to a worse payoff. 

\section{Generalized Model}

We now consider the generalized model where there are three parameters
corresponding to each agent $i$, $\tau_i \leq \theta_i \leq \sigma_i$.
An agent considers a product if $\tau_i \leq \phi$, adopts a product
if $\theta_i \leq \phi$, and advertises it to their friends if
$\sigma_i \leq \phi$. 
If an agent is exposed to a product 
but $\phi < \tau_i$, then the payoff associated with the agent is 0. 
If $\phi \in [\tau_i, \theta_i)$, then the agent rejects the product, for
a payoff of $-q < 0$. 
As before, there is a payoff of $p > 0$ if the
agent accepts the product; however, the agent only advertises the
product to its neighbors after accepting if $\phi \geq \sigma_i$.
We therefore have four types of agents:

\begin{itemize}
\item Type I: Agents for which $\phi < \tau_i$,
\item Type II: Agents for which $\phi \in [\tau_i, \theta_i)$,
\item Type III: Agents for which $\phi \in [\theta_i, \sigma_i)$, and
\item Type IV: Agents for which $\phi \geq \sigma_i$. 
\end{itemize}

We denote the set of all agents of Type I by $T_1$ (and likewise for the other types). The basic model above is the special case where $\tau_i = 0$ and $\theta_i = \sigma_i$ for all $i \in V$. In this instance, we only have agents of Types II and IV. 

\begin{lm}
Given any seed set $S$, we note: 
\begin{enumerate}
\item $\pi(S) = \pi(S \cup T_1)$, 
\item $\pi(S) \leq \pi(S \cup T_3)$.
\end{enumerate}
\label{lm:types}
\end{lm}

\begin{proof}
These follow from the observation that: 
\begin{enumerate}
\item Agents of Type I are those that do not look at the product since $\phi$ is below their threshold $\tau_i$, and thus do not affect the payoff function when added to any seed set. 
\item Agents of Type III are those for which $\phi \geq \theta_i$, and therefore they accept the product, but do not advertise it to their friends. Therefore, adding such an agent to the seed set increases the payoff by exactly $p > 0$ per agent added.
\end{enumerate}
\end{proof}

In the simplest case, we have $ \tau_i = \theta_i$ and $\sigma_i = 1$,
such that agents will be either of Type I or III, and the optimal
seed set is precisely $S = T_3$. That is, agents only view a product
if they are going to accept it and they never advertise it to their
neighbors, and there is no cascade triggered as a result of seeding
agents.

When this is not the case, we will note that the results in the previous section can be adapted to this setting to find an optimal seed set efficiently. Given a network $G$ in this generalized setting, consider a corresponding network $G' = (V', E')$, which is the subgraph of $G$ consisting of agents of only Types II and IV. We can then apply the algorithm in the previous section to this subgraph $G'$ to find an optimal seed set. We claim that the union of this with agents of Type III yields an optimal seed set in $G$. 

\begin{thm}
Given a product with a value $\phi$ and a network $G$ with agents of parameters $\tau_i, \theta_i,$ and $\sigma_i$, there is a polynomial-time algorithm for finding an optimal seed set to maximize the payoff function. 
\end{thm}

\begin{proof}
Given such a graph $G$, and a corresponding subgraph $G'$ with an optimal seed set $S'$, we argue that $S' \cup T_3$ is an optimal seed set in $G$. 

For the sake of contradiction, suppose $S$ is an optimal seed set in
$G$, such that $\pi(S) > \pi(S' \cup T_3)$. By Lemma \ref{lm:types}, we can
assume that $S$ does not include any agent of Type I and includes all
agents of Type III. 
This assumption implies $\pi( S \backslash T_3) > \pi(S')$.
But since $S \backslash T_3 \subseteq G'$, this contradicts the
optimality of $S'$.

\end{proof}

Returning to the implications for the Groupon effect, 
we note that a firm can efficiently maximize its payoff over the choice
of {\em both} the seed set and $\phi$.  Suppose the firm
knows the parameters $\tau_i, \theta_i, \sigma_i$ for each of the
agents $i$. Given $n$ agents, these values divide up the unit interval
into at most $4n$ subintervals $I_j$. It is easy to see that the
payoff depends only on which subinterval $\phi$ is contained in, but
does not vary within a subinterval.  Earlier, we saw the payoff need
not be monotone in $\phi$; but by trying values of $\phi$ in each of the $4n$
subintervals, the firm can determine a value of $\phi$ and a seed set
that maximize its payoff.

\section{Budgeted Seeding is Hard}\label{budgeted}

\def\payoffvar{\pi}

In this section, we show that the seeding problem, even for the
initial model, is NP-hard if we consider the case where there is a
budget $k$ for the size of the seed set and we want to find $S$
subject to the constraint that $|S| \leq k$. In the traditional
influence maximization literature, we leverage properties of the
payoff function such as its submodularity or supermodularity to give
algorithms that find optimal or near-optimal seed sets. The payoff
function here, however, is neither submodular nor supermodular.

\begin{eg}\label{counter2}

Take the network shown in figure below. Set $p = q = 1$: 
\begin{figure}[!ht]\label{fig:countereg2}
\centering
\begin{tikzpicture}
  [scale=1.25,auto=left,every node/.style={circle,fill=gray!20}]
  \node[fill=blue!50] (n1) at (1, 0) {1};
  \node[fill=blue!50] (n2) at (0.25, 1) {2};
  \node[fill=blue!50] (n3) at (1, 2) {3};
  \node[fill=red!50] (n4) at (2, 0.5) {4};
  \node[fill=red!50] (n5) at (2, 1.5) {5};
  \node[fill=blue!50] (n6) at (3, 1.5) {6};
  \node[fill=blue!50] (n7) at (3, 0.5) {7}; 
  \foreach \from/\to in {n1/n2, n2/n3, n1/n4, n2/n4,n1/n5,n2/n5,n3/n4,n3/n5,n5/n6,n6/n4,n5/n6,n1/n3,n4/n7,n5/n7,n6/n7}
    \draw (\from) -- (\to);
\end{tikzpicture}
\caption{Network where nodes in blue accept have $\theta_i \leq \phi$ and those in red have $\theta_i > \phi$.}
\end{figure}
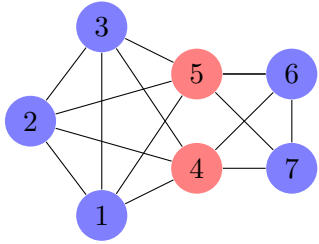

Supermodularity states that given $S' \subseteq S$ and $x \notin S$, $$f(S' \cup \{x\}) - f(S') \leq f(S \cup \{x\}) - f(S).$$ 

Set $S' = \{1 \}$, $S = \{1, 6\}$, and let $x$ be node $7$. Then, supermodularity would give: 
\begin{align*}
\pi(\{1, 7\}) - \pi(1) & \leq \pi \{1, 6, 7\} - \pi\{1, 6\}\\
2& \leq 0
\end{align*}

Submodularity states that, 
$$f(S' \cup \{x\}) - f(S') \geq f(S \cup \{x\}) - f(S).$$

A counterexample to this is obtained by setting $S' = \emptyset$, $S  = \{1\}$ and $ x = \{7\}$.

\end{eg}

  Here, we show an even stronger hardness
result: it is NP-hard to decide if there is a (budgeted) set yielding
positive payoff.  Since it is NP-hard to tell whether the optimum
in any instance is positive or negative, it is therefore also
NP-hard to provide an approximation algorithm with any 
multiplicative guarantee --- a sharp contrast with the 
multiplicative approximation guarantees available for
budgeted problems in more traditional influence maximization settings.

\begin{thm}
The decision problem of whether there exists a seed set $S$ with $|S| \leq k$ 
and $\payoffvar(S) > 0$ is NP-complete. 
\end{thm}

\begin{proof}

We will prove this using a reduction from the NP-complete
Clique problem on $d$-regular graphs: given a $d$-regular graph
$G$ and a number $k$, the question is to determine whether there
exists a {\em $k$-clique} --- a set of $k$ nodes that are all mutually adjacent.
(We also require $d \geq k$.)

We will reduce an instance of $k$-clique on $d$-regular graphs
to an instance of the decision version
of our budgeted seed set problem as follows.  Given such 
an instance of Clique specified by a $d$-regular graph $G$ and
a number $k$,
we construct an instance of the budgeted seed set
problem on a new graph $G'$ obtained from $G$ as follows: we replace each 
$(i, k) \in E$, with two new edges $(i, j), (j, k)$, where $j$ is a new node
introduced by {\em subdividing} $(i,k)$.
Let $V$ be the set of nodes originally in $G$, and $V'$ the set of
nodes introduced by subdividing.
In the seed set instance on $G'$, we define $\theta_i$ and $\phi$
such that $\theta_i < \phi$ for all $i \in V$ and $\theta_i > \phi$ for
all $i \in V'$.
We define the payoff coefficients $p, q$ 
by $q = 1$ and $p = d - (k-1)/2 + \epsilon$ for some
$0 < \epsilon < 1/n^2$.

We will show that $G$ has a $k$-clique if and only if there is a seed set
of size at most $k$ in $G'$ with positive payoff.

First, suppose that $S$ is a set of $k$ nodes in $G$ that are all
mutually adjacent, and consider the corresponding set of nodes $S$ in $G'$.
As a seed set, $S$ has $k$ accepting nodes and $kd - {k \choose 2}$
rejecting neighbors, since $G$ is $d$-regular but the nodes on
the ${k \choose 2}$ subdivided edges are double-counted.
Thus the payoff from $S$ is
$kp - (kd - {k \choose 2}) = k(p - d + (k-1)/2)$, which is positive
by our choice of $p$.

For the converse, suppose $S$ has size $k' \leq k$ and
has positive payoff in $G'$.
Since the seed set consists entirely of accepting neighbors 
(any others can be omitted without decreasing the payoff),
$S \subseteq V$, and hence so the neighbors of $S$ reject the product. 
If $S$ induces $\ell$ edges in $G$, then the payoff from $S$
includes a negative term from each neighbor, with the nodes on
the $\ell$ subdivided edges double-counted, so the payoff is
$k'p - (k'd - \ell)$.
If $|S| = k' < k$, then since $\ell \leq {k' \choose 2}$,
the payoff is at most
$k'p - (k'd - {k' \choose 2}) = k'(p - d + (k'-1)/2)$, which
is negative by our choice of $p$.
If $|S| = k$ and $\ell \leq {k \choose 2} - 1$, then the payoff is at most
$kp - (kd - ({k \choose 2} - 1)) = k(p - d + (k-1)/2 - 1/k)$,
which again is negative by our choice of $p$.
Thus it must be that $|S| = k$ and $\ell = {k \choose 2}$, so
$S$ induces a $k$-clique as required.
\end{proof}

\section{Experimental Results}

In this section, we present some computational results using datasets
obtained from SNAP (Stanford Network Analysis Project). In particular,
we consider an email network from a European
research institution \cite{email2,email1} and a text message network
from a social-networking platform at UC-Irvine \cite{text}.

The former is a directed network of emails sent between employees
over an 803-day period, with $986$ nodes and $24929$ directed edges.
The latter is a directed network of
text messages sent between students through an online social network
at UC Irvine over a 193-day period, with $1899$ nodes
and $20296$ directed edges. In both networks,
we use the edge $(i,j)$ to indicate that $i$ sent at least one
email or text to node $j$ over the time period considered.



For both of these networks, we present results corresponding to the
general model. We consider 100 evenly-spaced values of $\phi$ in $[0,
1]$ and compare the seed set obtained by our algorithm with some
natural baselines. The parameters 
$\tau_i \leq \theta_i \leq \sigma_i$ for each agent 
are chosen as follows: we draw three numbers independently
from an underlying distribution (we analyze both the
uniform distribution on $[0,1]$ and the 
Gaussian distribution with mean $0.5$ and standard deviation $0.1$);
we then sort these three numbers in non-decreasing order and
set them to be $\tau_i, \theta_i,$ and $\sigma_i$ respectively.


For each $\phi$ we run 100 trials and
present the average payoff. The average time to run one simulation is
$0.915$ seconds for the text network and $0.454$ seconds for the email
network.  This includes the time to read the data and assemble
the network; the average time spent only on computing the min-cut for the
corresponding network is $0.052$ and $0.054$ seconds respectively.

For each of these figures, we give a natural upper-bound which is the
number of agents such that $\theta_i < \phi$. This includes agents of
Type III and IV. We note that the seed set obtained by our algorithm
often gives a payoff close to this upper bound. We compare this to 
two natural baselines: the first sets agents of Type III to be the seed set 
and the second sets agents of both Types III and IV to be the seed set. 
We show that the first baseline performs well for lower values of $\phi$, 
where the second baseline underperforms significantly; and, the second 
picks up performance significantly for higher values of $\phi$ while the 
first baseline suffers. The seed set obtained by our algorithm, on the 
other hand, outperforms both baselines by a notable margin for moderate 
values of $\phi$. This gap in performance corresponds to the overexposure 
effect in our models.

\begin{figure*}[!ht]\label{fig:unifgen}
\centering
\includegraphics[width = 7.25cm,height=5.75cm]{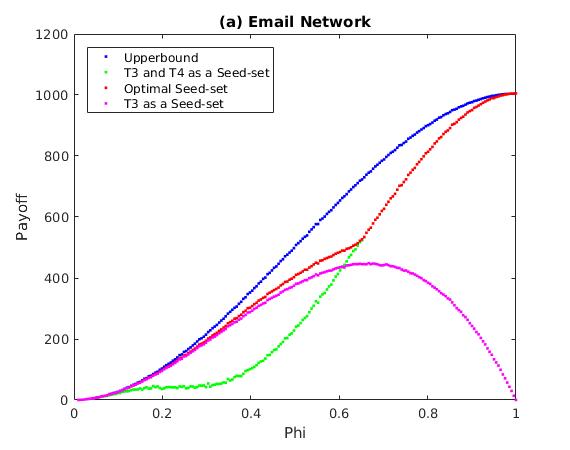}
\includegraphics[width = 7.25cm,height=5.75cm]{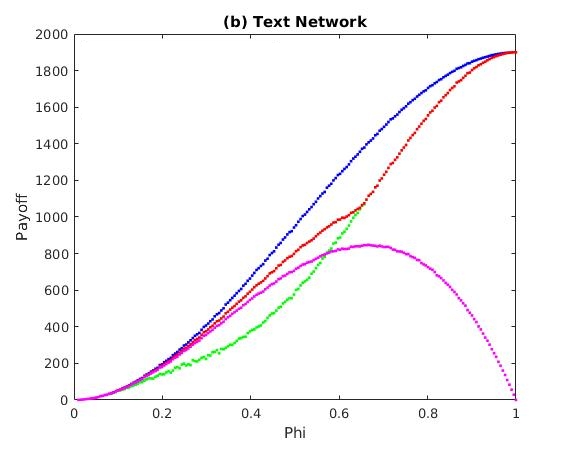}
\caption{Payoff as a function of product quality $\phi$ for (a) email and (b) text networks for the general model where $\tau_i, \theta_i, \sigma_i$ are chosen from the uniform distribution on $[0, 1]$ . The blue curve represents the natural upper bound of the total number of agents with $\theta_i < \phi$.  The green line corresponds to the payoff obtained by seeding agents whose $\tau_i \leq \phi$. The purple line represents the number of agents of Type III, where $\phi \in [\theta_i,\sigma_i)$ and the red line corresponds to the seed set selected by our algorithm. The difference between the red and green curve captures the Groupon Effect of overexposure.}
\end{figure*}

\begin{figure*}[!ht]\label{fig:gaussgen}
\centering
\includegraphics[width = 7.25cm,height=5.75cm]{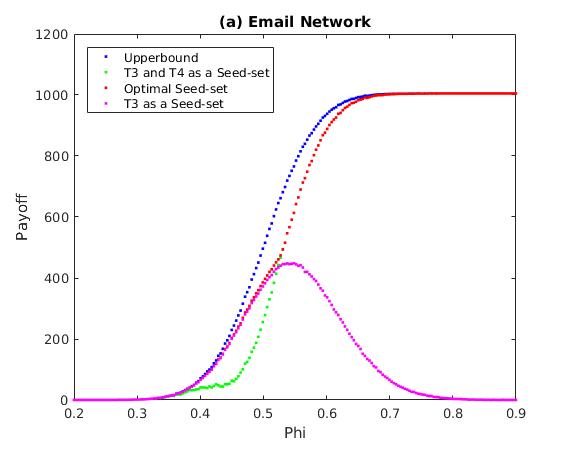}
\includegraphics[width = 7.25cm,height=5.75cm]{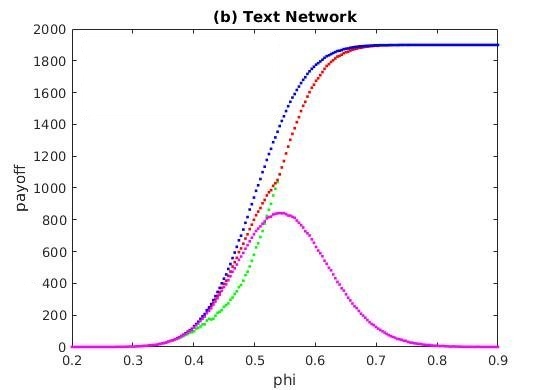}
\caption{These plots give payoff for (a) email and (b) text network for the general model with parameters chosen from the Gaussian distribution with mean $0.5$ and standard deviation $0.1$ for 100 evenly-spaced $\phi$ values in $[0, 1]$.}
\end{figure*}

Each of the figures show that the seed set chosen by our algorithm
outperforms these natural baselines. In Figure 2, we note for $\phi <
2/3$ the optimal seed set obtained through our algorithm is close to
picking only agents of Type III. Adding on agents of Type IV performs
worse than both seeding just agents of Type III or the optimal
seed set. This changes for $\phi$ values over $2/3$. Here, the number
of agents of Type III drops, and thus the payoff obtained by seeding
agents of Type III drops with it. 
On the other hand, the seed set consisting of all
agents of Types III and IV picks up performance, 
coming close to the 
optimal seed set for $\phi \approx 0.7$. This behavior appears in 
both networks.

\section{Further Related Work}

As noted in the introduction, our work --- through its focus on
selecting a {\em seed set} of nodes with which to start a cascade ---
follows the motivation underlying the line of theoretical work on 
influence maximization \cite{dr,kkt,rd}. There has been some 
theoretical work showing the counter-intuitive outcome where 
increased effort results in a less successful spread. An example is \cite{sela}, 
where they show that due to the separation of the infection and viral 
stage, there are cases where an increased effort can result in a lower 
rate of spread. A related line of work has made use of rich datasets on 
digital friend-to-friend recommendations on e-commerce sites
to analyze the flow of product recommendations through an
underlying social network \cite{leskovec-ec06j}. Further work has
experimentally explored influence strategies, with individuals
either immediately broadcasting their product adoption to 
their social network, or selecting individuals to recommend
the product to~\cite{aral2011creating}.

The consequences of negative consumer reactions have been
analyzed in a range of different domains.
In the introduction we noted examples involving
Groupon \cite{group}, book prizes \cite{good}, and 
on-line media collections \cite{aizen-batting}.
Although experimentally introduced negative ratings tend to be
compensated for in later reviews, positive reviews 
can lead to herding effects ~\cite{muchnik2013social}.
There has also been research seeking to quantify 
the economic impact of negative ratings, in contexts
ranging from seller reputations in on-line auctions
\cite{bajari-auction-reputation,resnick-value-reputation}
to on-line product reviews \cite{pang-lee-sentiment-book}.
This work has been consistent in ascribing non-trivial
economic consequences to negative consumer impressions
and their articulation through on-line ratings and reviews. 
Recent work has also considered the rate at which social-media content
receives ``likes'' as a fraction of its total impressions,
for quantifying a social media audience's response
to cascading content \cite{rotabi-www17-audience}.

The literature on pricing goods with network effects is another domain
that has developed models in which consumers are
heterogeneous in their response to diffusing content.
The underlying models are different from what we pursue here;
a canonical structure in the literature on pricing with network effects
is a set of consumers with different levels of willingness to pay for a product
\cite{katz-shapiro-net-ext}.
This willingness to pay can change as the product becomes more popular;
a line of work has thus considered how a product with network effects
can be priced adaptively over time as it diffuses through the network
\cite{arthur-pricing-social,hartline-pricing-social}.
The variation in willingness to pay can be viewed as a type of
``criticality,'' with some consumers evaluating products 
more strictly and others less strictly.
But a key contrast with our work is that highly critical 
individuals in these pricing models do not generally confer a negative
payoff when they refuse to purchase an item.

\section{Conclusion}

Theoretical models of viral marketing in social networks have generally 
used the assumption 
that all exposures to a product are beneficial to the firm conducting
the marketing.  
A separate line of empirical research in marketing, however, provides a more
complex picture, in which different potential customers may have
either positive or negative reactions to a product, and it can be
a mistake to pursue a strategy 
that elicits too many negative reactions from potential customers.

In this work, we have 
proposed a new set of theoretical models for viral marketing,
by taking into account these types of overexposure effects.
Our models make it possible to consider the optimization trade-offs
that arise from trying to reach a large set of positively inclined
potential customers while reducing the number of negatively inclined
potential customers who are reached in the process.
Even in the case where the marketer has no budget on the number
of people it can expose to the product, this tension between
positive and negative reactions leads to a non-trivial optimization problem.
We provide a polynomial-time algorithm for this problem, using
techniques from network flow, and we prove hardness for the case in 
which a budget constraint is added to the problem formulation.
Computational experiments show how our polynomial-time algorithm
yields strong results on network data. 

Our framework suggests many directions for future work.
It would be interesting to integrate the role of negative
payoffs in our model here with other technical components 
that are familiar from the literature on influence maximization,
particularly the use of richer (and potentially probabilistic)
functions governing the spread from one participant in the network
to another.
For example, when nodes have non-trivial thresholds for adoption ---
requiring both that they evaluate the product positively and also
that they have heard about it from at least $k$ other people, 
for some $k > 1$ --- how significantly do the structures of optimal
solutions change?

It will also be interesting to develop richer formalisms for the process by
which positive and negative reactions arise when potential customers
are exposed to good or bad products.
With such extended formalisms we can more fully bring together
considerations of overexposure and reputational costs into the
literature on network-based marketing.

\section*{Acknowledgements}
The first author was supported in part by a Google scholarship, a Facebook 
scholarship, and a Simons Investigator Award and the third author was supported 
in part by a Simons Investigator Award, an ARO MURI grant, a Google Research 
Grant, and a Facebook Faculty Research Grant.

\bibliographystyle{plain}
\bibliography{sigproc}

\end{document}